\documentclass[submission,copyright,creativecommons]{eptcs}
\providecommand{\event}{MBT 2014} 
\usepackage{breakurl}             

\usepackage{amssymb,amsmath,amsthm,amsfonts}
\usepackage{subfigure,samp,mathpartir}
\usepackage[T1]{fontenc}

\newtheorem{lemma}{Lemma}
\newtheorem{definition}{Definition}
\newtheorem{example}{Example}

\newtheorem{theorem}{Theorem}

\usepackage{tikz}
\usetikzlibrary{matrix,arrows}
\usetikzlibrary{shapes,arrows,automata}
\usetikzlibrary{shapes.symbols}
\usetikzlibrary{shapes.misc,shapes.multipart}
\usetikzlibrary{calc}
\usetikzlibrary{decorations.pathreplacing}

\title{Spinal Test Suites for Software Product Lines}
\author{Harsh Beohar
\institute{Center for Research on Embedded Systems\\ Halmstad University, Sweden}
\email{harsh.beohar@hh.se}
\and
Mohammad Reza Mousavi
\institute{Center for Research on Embedded Systems\\ Halmstad University, Sweden}
\email{m.r.mousavi@hh.se}
}
\def\titlerunning{Spinal Test Suites for Software Product Lines}
\def\authorrunning{H. Beohar \& M.R. Mousavi}
\begin{document}
\maketitle

\begin{abstract}
A major challenge in testing software product lines is efficiency.
In particular, testing a product line should take less effort than testing each and every product individually.
We address this issue in the context of input-output conformance testing, which is a formal theory of model-based testing.
We extend the notion of conformance testing on input-output featured transition systems with the novel concept of spinal test suites. We show how this concept dispenses with retesting the common behavior among different, but similar, products of a software product line.
\end{abstract}

\section{Introduction}

\subsection{Motivation}
Testing and debugging are labor-intensive parts of software development.
In particular, testing a software product line is extremely time- and resource-consuming due to the various configurations of products  that are derivable from the product line.
In order to manage the complexity, the test process of a software product line must be efficiently coordinated: common features ought to be tested once and for all and only specific variation points of various configurations should be tested separately.

Model-based testing is an approach to structure the test process by  exploiting test models.
Input-output conformance testing (ioco) \cite{Tretmans08} is a formalization of model-based testing employing input-output labeled transition systems as models.
In the past, we extended the formal definition ioco to the setting of software product lines \cite{fioco-sac14}, by exploiting input-output featured transition systems.
In this paper, we define a theoretical framework, which serves as the first step towards an efficient discipline of conformance model-based testing for software product lines.

To this end, we define the notion of \emph{spinal test suite}, which allows one to test the common features once and for all, and subsequently, only focus on the specific features when moving from one product configuration to another.
We show that spinal test suites are exhaustive, i.e., reject each and every non-conforming implementation under test, when the implementation satisfies the \emph{orthogonality criterion}. This is a rather mild criterion, which implies that old features are not capable of disabling any enabled behavior from the new features on their own and without involving any interaction with the new feature's components.


\subsection{Running example}
To motivate various concepts throughout the paper,
we use the following running example.
Consider an informal description of a cruise controller, present in  contemporary cars. The purpose of a cruise controller is
to automatically maintain the speed of the car as specified by the
driver.
We denote the basic feature of a cruise controller by $\action{cc}$.
Cruise
controllers also have an optional feature, called  collision avoidance controller ($\action{cac}$), whose task is to
react to any obstacle detected ahead of the car within a danger zone. In case the collision avoidance
feature is included in a cruise controller and an obstacle is detected, the engine power is regulated using an emergency control algorithm.

\subsection{Organization}

The rest of this paper is structured as follows.
In Section \ref{sec::backg}, we recall the formal definitions regarding models, product derivation and conformance testing.
In Section \ref{sec:spinal}, we define the notion of spinal test suite, which is a compact test suite for the ``new'' features with respect to an already tested product (or a set of features).
In Section \ref{sec:exhaust}, we study the exhaustiveness of the spinal test suites: we show that spinal test suites are in general non-exhaustive, but this can be remedied by requiring mild conditions on the implementation under test.
In Section \ref{sec:related}, we sketch the context of this research.
In Section \ref{sec:conc}, we conclude the paper and outline the  direction of our ongoing research.

\section{Background}\label{sec::backg}
\subsection{Input-output featured transition systems}
Feature diagrams \cite{Kang90,Schobbens:2006} have been used to model variability constraints in SPLs using a graphical notation. However,
it is well known that feature diagrams only specify the structural aspects of
variability and they should be complemented with other models in order to specify the behavioral aspects \cite{Classen:2012:fts}. To this end, we describe
the behavior of a software product line using an input-output featured
transition system (IOFTS) \cite{fioco-sac14}, defined and explained below.

Let $F$ be the set of features (extracted from a feature diagram) and $\B=\{\top,\bot\}$ be the set of Boolean constants; we denote by  $\B(F)$ the set of all
propositional formulae generated by interpreting the elements of the set $F$ as propositional
variables. For instance, in our running example, formula
$\action {cc}\wedge\neg\action {cac}$ asserts the presence of cruise controller and the absence of
collision avoidance controller. We let $\phi,\phi'$ range over the set $\B(F)$.

\begin{definition}\label{def:iofts}
A \emph{input-output featured transition system}
(IOFTS) is a 6-tuple $(S,s,\tact,F,T,\Lambda)$, where
\begin{enumerate}
    \item $S$ is the set of \emph{states},
    \item $s\in S$ is the \emph{initial state},
    \item $\tact=\inp\uplus\out \uplus \set\tau$ is the set of \emph{actions}, where $\inp$ and $\out$ are disjoint sets of \emph{input} and \emph{output} actions, respectively, and $\tau$ is the silent (internal) action,
    \item $F$ is a set of \emph{features},
    \item $T\subseteq S\times \tact \times \B(F) \times S$ is the \emph{transition relation} satisfying the following condition (for every $s_1,s_2\in S,a\in\tact,\phi,\phi'\in\B(F)$):
        \[(s_1,a,\phi,s_2)\in T \wedge (s_1,a,\phi',s_2) \in T\Rightarrow \phi=\phi',\footnote{Here, by $\phi=\phi'$ we assert that $\phi$ and $\phi'$ are syntactically equivalent.}\]
    \item $\Lambda\subseteq\set{\lambda:F\rightarrow \B}$ is a set of \emph{product configurations}.
\end{enumerate}
\end{definition}

We write $s\stepa a \phi s'$ to denote an element $(s,a,\phi,s')\in T$ and drop the subscript $\phi$ whenever it is clear from the context. Graphically, we denote the initial state of an IOFTS by an incoming arrow with no source state and we refer to an IOFTS by its initial state. Following the standard notation, we denote the \emph{reachability} relation by $\steps{}\subseteq S\times \fseq \act \times S$, which is inductively defined as follows:
\[\frac{}{s\steps \eseq s} \quad \frac{s\steps \sigma s', s' \step \tau s''}{s\steps {\sigma } s''} \quad \frac{s\steps \sigma s', s' \step a s'',a\neq \tau}{s\steps {\sigma a} s''}.\]
Furthermore, the set of \emph{reachable} states from a state $s$ is denoted by $\reach s =\set{s'\mid \exists_{\sigma}\ s\steps\sigma s'}$.

\begin{example}
Consider the IOFTS of a cruise controller, drawn in Figure~\ref{fig:cc}, where inputs
and outputs are prefixed with symbols $?$ and $!$, respectively.
(Note that $?$ and $!$ are \emph{not part} of the action names and are left out when the type of the action is irrelevant or clear from the context.)
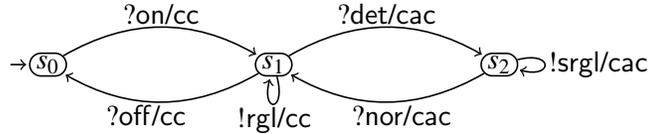
\begin{figure}[htb]
\centering
\tikzset{
  state/.style={
    rounded rectangle,
    draw=black,
  }
}
\begin{tikzpicture}[->,shorten >=1pt,auto, semithick,inner sep=1pt]
\node[state] (s1) {$s_0$};
\node[state] (s2) at ($(s1.center) + (3,0)$) {$s_1$};
\node[state] (s3) at ($(s2.center) + (3,0)$) {$s_2$};
\path[->]
    ($(s1.center)+(-0.5,0)$) edge (s1)
    (s1) edge[bend left] node {?\action{on}/\action{cc}} (s2)
    (s2) edge[bend left] node[xshift=-0.2cm] {?\action{off}/\action{cc}} (s1)
    (s2) edge[loop below] node {!\action{rgl}/\action{cc}} (s2)
    (s2) edge[bend left] node {?\action{det}/\action{cac}} (s3)
    (s3) edge[bend left] node[xshift=0.2cm] {?\action{nor}/\action{cac}} (s2)
    (s3) edge[loop right] node {!\action{srgl}/\action{cac}} (s3);
\end{tikzpicture}
\caption{IOFTS of the cruise controller.}
\label{fig:cc}
\end{figure}
The regulate action, indicated by $\action{rgl}$, regulates the engine power of the car when the cruise controller is activated.
Furthermore, when $\action{cac}$ is included in a product, some additional behavior may emerge. Namely, while the cruise controller is on, if an object is detected within a danger zone, then the cruise
controller regulates the engine power in a safe manner  denoted by $\action{srgl}$. When the sensor signals a normal state,
the cruise controller returns to the normal regulation regime.
(For a realistic case study of a cruise controller and its formal model, we refer to \cite{delangen12}.)
\end{example}

%

\subsection{Product derivation operators}
In \cite{fioco-sac14}, we introduced a family of product derivation operators (parameterized by feature constraints),
which project the behavior of an IOFTS into another IOFTS representing a selection of products (a product sub-line).

\begin{definition}\label{def:fts-proj}
Given a feature constraint $\phi$ and an IOFTS $(S,s,\tact,F,T,\Lambda)$, the projection operator $\Delta_\phi$ induces an IOFTS $(S',\prd\phi{s},\tdact,F,T',\Lambda')$, where
\begin{enumerate}
\item $S'=\set{\prd \phi {s'} \mid s'\in S}$ is the set of states,
\item $\prd \phi s$ is the initial state,
\item $\tdact=\tact\uplus\set\delta$ is the set of actions, where $\delta$ is the special action label modeling quiescence \cite{Tretmans08},
\item $T'$ is the smallest relation satisfying:
\begin{mathpar}
\irule{s\stepa a {\phi'} s' \\\\
\exists_{\lambda}\ (\lambda \in\Lambda \wedge \lambda\models(\phi\wedge\phi')) }{\prd \phi s \stepa a {\phi\wedge \phi'} \prd \phi {s'}}\label{rule:act}\\
\irule{\bar\Lambda=\set{\lambda\in\Lambda\mid \lambda\models\phi \wedge \mathbf{Q}(s,\lambda)}\quad
\bar\Lambda\neq\emptyset}{\prd \phi s \stepa \delta {\phi\wedge (\bigvee_{\lambda\in\bar\Lambda}\ \lambda)} \prd \phi s}\label{rule:quiet}
\end{mathpar}
where the predicate $\mathbf{Q}(s,\lambda)$ is defined as $$\forall_{s',a,\phi'}\ \big(s\stepa a {\phi'} s' \wedge a\in\out\cup\set\tau\big) \Rightarrow \lambda\not\models\phi'.$$
\item $\Lambda'=\set{\lambda\in \Lambda\mid \lambda\models\phi}$ is the set of product configurations.
\end{enumerate}
\end{definition}

In the above-given rules $\lambda \models \phi$, denotes that valuation $\lambda$ of features satisfies feature constraint $\phi$.
Intuitively, rule~\eqref{rule:act} describes the behavior of those valid products that satisfy the feature constraint $\phi$ in addition to the original annotation of the  transition emanating from $s$.
Rule~\eqref{rule:quiet} models quiescence (the absence of outputs and internal actions) from the state $\prd{\phi}s$.
Namely, it specifies that the projection with respect to $\phi$ is quiescent, when there exists a valid product $\lambda$ that satisfies $\phi$ and is quiescent, i.e., cannot perform any output or internal transition. Quiescence at state $s$ for a feature constraint $\lambda$ is formalized using the predicate
$\mathbf{Q}(s,\lambda)$, which states that from state $s$ there is no output or silent transition with a constraint satisfied by $\lambda$.  In the conclusion of the rule, a $\delta$ self-loop is specified and its constraint holds when $\phi$ holds and at least the feature constraint of one quiescent valid product holds.
This ability to observe the absence of outputs (through a timeout mechanism) is crucial in defining the input-output conformance relation between a specification and an implementation \cite{fioco-sac14}.

\begin{example}\label{ex:cruise-cons-spec}
Consider the feature constraint $\phi=\action{cc} \wedge\neg\action{cac}$. The IOFTS
generated by projecting the IOFTS of cruise controller (in Figure~\ref{fig:cc}) using feature constraint $\phi$
is depicted in Figure~\ref{fig:cc:proj}. As mentioned before, this represents the product that has the basic cruise controller functionality but does not contain collision avoidance controller.
\end{example}

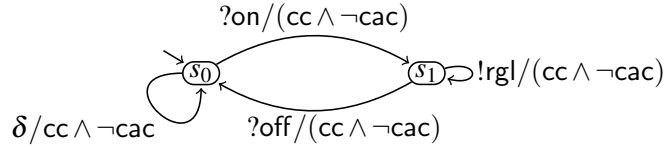
\begin{figure}[htb]
\centering
\tikzset{
  state/.style={
    rounded rectangle,
    draw=black,
  }
}
\begin{tikzpicture}[->,shorten >=1pt,auto, semithick,inner sep=1pt]
\node[state] (s1) {$s_0$};
\node[state] (s2) at ($(s1.center) + (3,0)$) {$s_1$};
\path[->]
    ($(s1.center)+(-0.5,0.3)$) edge (s1)
    (s1) edge[loop below, in=270, out=180, looseness=10] node[below left] {$\delta/\action{cc}\wedge\neg\action{cac}$} (s1)
    (s1) edge[bend left] node {$?\action{on}/(\action{cc}\wedge\neg\action{cac})$} (s2)
    (s2) edge[bend left] node[xshift=0.4cm]{$?\action{off}/(\action{cc}\wedge\neg\action{cac})$} (s1)
    (s2) edge[loop right] node {$!\action{rgl}/(\action{cc}\wedge\neg\action{cac})$} (s2);
\end{tikzpicture}
\caption{Cruise controller after projecting with feature constraint $\action{cc}\wedge\neg\action{cac}$.}
\label{fig:cc:proj}
\end{figure}

In the sequel, we use the phrase ``a feature specification $\prd\phi s$'' to refer to the following IOFTS: $$(\reach{\prd \phi s},\prd\phi s, \tdact,F,T,\Lambda).$$ 
We interpret the original IOFTS of Definition~\ref{def:iofts} as $\prd \top {s_0}$; this has the implicit advantage of always including quiescence in appropriate states.

\subsection{Input-output conformance}
The input-output conformance (ioco) testing theory \cite{Tretmans08} formalizes model-based testing in terms of a conformance relation between the states of a model (expressed as an input-output transition system) and an implementation under test (IUT). Note that the ioco theory is based on the \emph{testing assumption} that the behavior of the IUT can be expressed by an input-output transition system, which is unknown to the tester.

The conformance relation can be checked by constantly providing the SUT with inputs that are deemed relevant by the model and observing outputs from the SUT and comparing them with the possible outputs prescribed by the model. In the following, we recall such an \emph{extensional} definition of ioco, extended to software product lines in \cite{fioco-sac14}. An equivalent \emph{intensional} definition of ioco that relies on comparing the traces of the underlying IOFTS was also given in \cite{fioco-sac14}, but for the purpose of this paper  we only work with the extensional definition.
(After all, the extensional definition is the one that is supposed to be applied in practice.)

We begin with a notion of \emph{suspension traces} generated by an IOFTS. Informally, a suspension trace is a trace that may contain the action $\delta$ denoting quiescence \cite{Tretmans08}.
\begin{definition}
The set of \emph{suspension traces} of a feature specification $\prd\phi s$, denoted by $\straces{\prd\phi s}$ is defined as: $\set{\sigma\in\fseq{\act_{\delta}}\mid \exists_{s'}\ \prd\phi s \steps\sigma \prd \phi {s'}}$.
\end{definition}

For example, in the IOFTS of Example \ref{ex:cruise-cons-spec},
$\action{\delta}$$?\action{on}$$!\action{rgl}$ is a suspension trace emanating from the initial state $s_0$.
Next, we define the notion of test suite, which summarizes all possible test cases that can be generated from a feature specification.


\begin{definition}\label{def:testsuite}
The \emph{test suite} for an IOFTS $(\reach{\prd\phi s},\prd \phi s,\tdact,F,T,\Lambda)$, dennoted by $\mcal T(s,\phi)$,  is the IOFTS $(\X\cup\set{\pass,\fail},\X_0, \dact,F, T',\Lambda),$, where
\begin{enumerate}
    \item $\X=\big\{\big(\set{s'\mid \prd \phi s\steps \sigma \prd \phi {s'}},\sigma\big)\mid \sigma\in\straces{s}\big\}$ is the set of intermediate states and $\set{\pass,\fail}$ is the set of \emph{verdict states} \cite{Tretmans08},
    \item $\X_0=\set{(\set{s'\mid \prd \phi s\steps \eseq \prd \phi {s'}},\eseq)}$ is the initial state of the test suite,
    \item $\dact=\act\uplus\set\delta$ is the set of actions, and
    \item the transition relation $T'$ is defined as the smallest relation satisfying the following rules.
        \begin{mathpar}
            \irule{(X,\sigma),(Y,\sigma a)\in \X}{(X,\sigma) \stepa {a} {\phi} (Y,\sigma a)}\label{rule:tb-a-step}\qquad\qquad
            \irule{a\in\out\cup\set\delta\\\\ (X,\sigma) \stepa {a} \phi (Y,\sigma' )}{(X,\sigma) \stepa {a} \phi \pass}\label{rule:tb-pass}

            \irule{ a\in\out\cup\set\delta\\\\ (X,\sigma)\not\stepa{a}\phi\pass}{(X,\sigma) \stepa {a} \phi \fail}\label{rule:tb-fail}\qquad\qquad
            \irule{a\in\out\cup\set\delta}{\pass \stepa a \phi \pass \\\\ \fail \stepa a\phi \fail}\label{rule:tb-loop}
        \end{mathpar}
\end{enumerate}
\end{definition}
Intuitively, the test suite for a feature specification is an IOFTS (possibly with an infinite number of states), which contains all the possible test cases that can be generated from the feature specification. Rule~\eqref{rule:tb-a-step} states that if $X$ and $Y$ are nonempty sets of reachable states from $s$ (under feature restriction $\phi$) with the suspension traces $\sigma$ and $\sigma a$, respectively, then there exists a transition of the form $(X,\sigma) \stepa{a}\phi (Y,\sigma a)$ in the test suite.
Rules~\eqref{rule:tb-pass} and \eqref{rule:tb-fail} model, respectively, the successful and the unsuccessful observation of outputs and quiescence. Note that input actions are not included in rules \eqref{rule:tb-pass} and \eqref{rule:tb-fail} because the implementation is assumed to be input-enabled \cite{Tretmans08}; hence, they are already covered by rule~\eqref{rule:tb-a-step}. Rule~\eqref{rule:tb-loop} states that the verdict states contain a self-loop for each and every output action, as well as for  quiescence.

\begin{figure}[htb]
\centering
\tikzset{
  state/.style={
    rounded rectangle,
    draw=black,
  }
}
\tikzset{
  state1/.style={
    circle,
    draw=black,
    fill=black
  }
}
{
\begin{tikzpicture}[->,shorten >=1pt,auto, semithick,inner sep=1pt, minimum size=0pt]
\node[state] (s1) {\scriptsize$\set{s_0},\eseq$};
\node (s1a) at ($(s1.center)+(-2,0)$) {$\cdots$};
\node (s1f) at ($(s1.center)+(3,-1.5)$) {\scriptsize$\fail$};
\node[state] (s2) at ($(s1.center)+(0,-1.5)$) {\scriptsize$\set{s_1},\action{on}$};
\node[state] (s2l) at ($(s2.center)+(-3,0)$) {\scriptsize$\set{s_0},\action{on}\ \action{off}$};
\node (s2ld1) at ($(s2l.center)+(-2.5,0.5)$) {$\cdots$};
\node (s2ld2) at ($(s2l.center)+(-2.5,-0.5)$) {\scriptsize$\fail$};
\node[state] (s2ld3) at ($(s2l.center)+(0,-1.5)$) {\scriptsize$\set{s_1},\action{on}\ \action{off}\ \action{on}$};
\node[state] (s2ld4) at ($(s2ld3.center)+(0,-1.5)$) {\scriptsize$\set{s_2},\action{on}\ \action{off}\ \action{on}\ \action{det}$};
\node (s2ld5) at ($(s2ld4.center)+(0,-1.5)$) {$\cdots$};
\node[state] (s3) at ($(s2.center)+(0,-1.5)$) {\scriptsize$\set{s_2},\action{on}\ \action{det}$};
\node (s3d) at ($(s3.center)+(0,-1.5)$) {$\cdots$};
\node[state] (s4) at ($(s2.center)+(3,-1.5)$) {\scriptsize$\set{s_1},\action{on}\ \action{rgl}$};
\node[state] (s5) at ($(s4.center)+(0,-1.5)$) {\scriptsize$\set{s_2},\action{on}\ \action{rgl}\ \action{det}$};
\node (s5d) at ($(s5.center)+(0,-1.5)$) {$\cdots$};
\path[->]
    ($(s1.center)+(0,0.5)$) edge (s1)
    (s1) edge node[above,sloped] {\scriptsize\action{rgl,srgl}} (s1f)
    (s1) edge node[above] {\scriptsize\action{\delta}} (s1a)
    (s2) edge node[above] {\scriptsize\action{\delta,srgl}} (s1f)
    (s1) edge node {\scriptsize\action{on}} (s2)
    (s2) edge node[above] {\scriptsize\action{off}} (s2l)
    (s2l) edge node[above] {\scriptsize\action{\delta}} (s2ld1)
    (s2l) edge node[above,sloped] {\scriptsize\action{rgl,srgl}} (s2ld2)
    (s2l) edge node[left] {\scriptsize\action{on}} (s2ld3)
    (s2ld3) edge node[left] {\scriptsize\action{det}} (s2ld4)
    (s2ld4) edge node[left] {\scriptsize\action{srgl}} (s2ld5)
    (s2) edge node[left] {\scriptsize\action{det}} (s3)
    (s3) edge node[left] {\scriptsize\action{srgl}} (s3d)
    (s2) edge node[above,sloped] {\scriptsize\action{rgl}} (s4)
    (s4) edge node[left] {\scriptsize\action{det}} (s5)
    (s5) edge node[left] {\scriptsize\action{srgl}} (s5d);
\end{tikzpicture}
}
\caption{The test suite of the cruise controller example.}
\label{fig:tsuite-cc}
\end{figure}
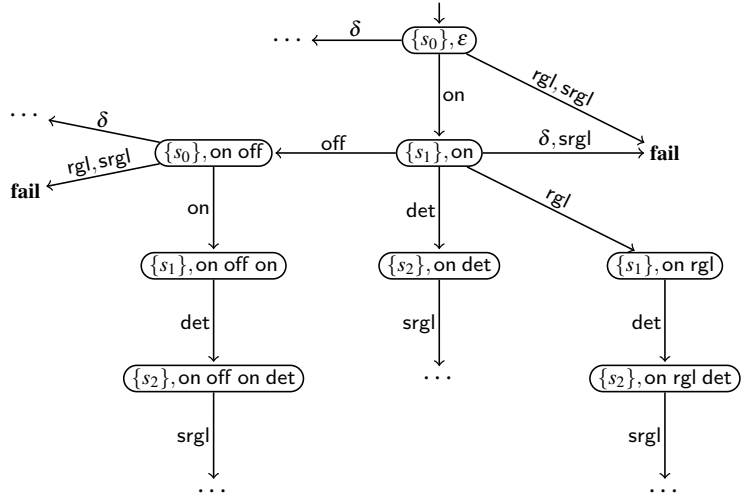

\begin{example}
The test suite for the IOFTS of Example~\ref{ex:cruise-cons-spec} is (partially) depicted in Figure~\ref{fig:tsuite-cc}.
\end{example}
A reader familiar with the original ioco theory \cite{Tretmans08} will immediately notice that our definition of a test suite (Definition~\ref{def:testsuite}) is nonstandard. In particular, a test suite is defined as a set of test cases (i.e., input-output transition systems with certain restrictions) with finite number of states in \cite{Tretmans08}; whereas we represent a test suite by an IOFTS, possibly with an infinite number of states. To this end, we define a test case to be a finite projection of a test-suite with the additional restriction that at each
moment of time at most one input can be fed into the system under test (see \cite{fioco-sac14} for a formal definition). As a result, our test cases are structurally similar to Tretmans' formulation of the test cases, by which we mean that:
\begin{itemize}
    \item a test case is always deterministic,
    \item a test case is always input enabled, and
    \item a test case has no cycles except those in the verdict states $\pass$ and $\fail$.
\end{itemize}
Another notable difference, that is key to define the concepts of Section~\ref{sec:spinal}, is that states of a test suite (or test case) carry some mathematical structure, whereas the states of a test case in \cite{Tretmans08} are abstract and carry no structure.

Next, we define a \emph{synchronous observation} operator $\parallel$ that allows us to model a test run on an implementation (cf. \cite{Tretmans08}). This is defined over a test suite and an IOFTS (the intended implementation) as follows. (Note that the calligraphic letters $\mcal X,\mcal Y$ in the following rules range over the states of a test suite.)
\begin{mathpar}
\irule{\mcal X \step a \mcal Y \quad \prd \phi s \step a \prd \phi {s'} \quad a\in\act_{\delta}}{\mcal X \parallel \prd \phi s \stepa a \top \mcal Y \parallel \prd \phi {s'}}

\irule{\prd \phi s \step \tau \prd \phi {s'} }{\mcal X \parallel \prd \phi s \stepa \tau \top \mcal X \parallel \prd \phi {s'}}

\end{mathpar}
Having defined the notion of synchronous observation, we can now define what it means for a feature specification to pass s test suite. Informally, a feature specification passes a test suite  if and only if no trace of the synchronous observation of the test suite and the feature specification leads to the $\fail$ verdict state.
\begin{definition}
Let $\X_0$ be the initial state of a test suite $\mcal T(s,\phi)$. A feature specification $\prd{\phi'}{s'}$ \emph{passes} the test suite $\mcal T(s,\phi)$ iff
$$\forall_{\sigma\in\fseq{\act_\delta},\mcal X,s''}\ \ \X_0\parallel \prd {\phi'}{s'} \steps \sigma \mcal X\parallel \prd {\phi'} {s''} \ \Rightarrow \ \mcal X\neq\fail.$$
The implementation $\prd{\phi'}{s'}$ \emph{conforms} to the specification $\prd\phi s$ iff $\prd{\phi'}{s'}$ passes the test suite $\mcal T(s,\phi)$.
\end{definition}

\section{Spinal test suite}\label{sec:spinal}
As mentioned in the introduction, one of the challenges in testing a software product line is to minimize the test effort. The idea pursued in this section is to organize the test process of a product line incrementally.
This is achieved by reusing the test results of an already tested product to test a product with similar features, thereby dispensing with the test cases targeted at  the common features.
To this end, we introduce the notion of \emph{spinal test suite}, which prunes away the behavior of a specified set of features from an \emph{abstract} test suite $\mcal T(s,\phi)$ with respect to a \emph{concrete} test suite $\mcal T(s,\lambda)$ of the already tested product $\lambda$; the spinal test suite is only defined when $\lambda$ is  valid w.r.t.\ $\phi$, i.e., $\lambda\models\phi$. The latter constraint means that the concrete product builds upon the already-tested features in the abstract test suite.

Notably, which behavior has to be pruned from an abstract test suite is crucial in defining a spinal test suite. One way to address this situation is \emph{by allowing only those reachable states in the abstract test suite from which a new behavior relative to the already tested product emanates}. However, without any formal justification, we claim that such a strategy will not reduce the effort to test new behavior with respect to the already tested product.

For example, consider the test suite depicted in Figure~\ref{fig:tsuite-cc} and suppose we have already tested the cruise controller without collision avoidance feature and now are interested in the correct implementation of the collision avoidance feature. By following the aforementioned strategy of pruning, none of the following states $(\set{s_1},\action{on}),(\set{s_1},\action{on}\ \action{off}\ \action{on}), \cdots$ will be removed because the event $\action{det}$ is enabled from each of these states. On the other hand, since we know that cruise controller without collision avoidance feature was already tested, it is safe to consider the new suspension traces (or testing experiments) from only one state in $\{(\set{s_1},\action{on}),(\set{s_1},\action{on}\ \action{off}\ \action{on}), \cdots\}$.

\begin{definition}
Let $\X_0$ be the initial state of a test suite $\mcal T(s,\phi)$. A path $\X_0 \steps\sigma (X,\sigma)$ is a \emph{spine} of a path $\X_0 \steps{\sigma'} (X,\sigma')$, denoted by $\sigma\spine \sigma'$, when $\sigma$ is a sub-trace of $\sigma'$ (obtained by removing zero or more action from $\sigma'$) and no two states visited during the trace $\sigma$ have the same $X$-component; this is formalized by the predicate $\bt(X,\sigma)$, defined below:
\[\forall_{\sigma_1,\sigma_2,\sigma_3,Y,Z}\ \big(\X_0 \steps{\sigma_1} (Y,\sigma_1) \steps{\sigma_2} (Z,\sigma_2) \steps{\sigma_3} (X,\sigma) \land \sigma_2 \neq \varepsilon \land \sigma=\sigma_1\sigma_2\sigma_3\big) \Rightarrow Y\neq Z.\]
Furthermore, we let $\bt(\X_0)=\top$.
%
%
\end{definition}

\begin{example}\label{ex:phi-lambda}
Recall the feature specification given $\prd \phi {s_0}$ in Example~\ref{ex:cruise-cons-spec}, where $\phi=\action{cc}\wedge\neg\action{cac}$. Since collision avoidance controller is an optional feature, we know that there exists a product configuration $\lambda$ with $\lambda(\action {cc})=\top$ and $\lambda(\action{cac})=\bot$. Then, the path labelled ``$\action{on}$'' (in the test suite drawn in Figure~\ref{fig:tsuite-cc}) is a spine of the path labelled ``$\action{on}\ \action{off}\ \action{on}$'' because they both reach to a common $X$-component $\set{s_1}$ in the test suite and $\bt(\set{s_1},\action{on})=\top$.
\end{example}
%
\begin{definition}
Let $(\X\cup\set{\pass,\fail},\X_0, \act_\delta,F,T,\Lambda)$ be a test suite $\mcal T(s,\phi)$ and let $\lambda$ be a product such that $\lambda\models\phi$. Then a \emph{spinal test suite} with respect to a product $\lambda$, denoted by $\sres$, is an IOFTS $(\X\cup\set{\pass,\fail},\X_0,\act_\delta,F,T',\Lambda')$, where
\begin{enumerate}
    \item The set of non-verdict states $\X$ is defined as $\X'\cup\X''$, where
        \begin{align*}
        \X'=&\ \set{(X,\sigma)\in\X^\phi_s \mid \sigma\in\straces{\prd\lambda s} \wedge \bt(X,\sigma)}\\
        \X''=&\ \{(Y,\sigma a \sigma')\in\X^\phi_s \mid \new (\sigma, a)\wedge \exists_{X}\ (X,\sigma)\in\X' \wedge (X,\sigma) \steps{a \sigma'} (Y,\sigma a \sigma')\}.
        \end{align*}
        where, $\new(\sigma, a)\ \Leftrightarrow\ \sigma\in\straces{\prd\lambda s} \wedge \exists_{s',s''}\ \prd\phi s\steps{\sigma} \prd\phi {s'} \stepa a {\phi'}\prd \phi{s''} \wedge \lambda\not\models\phi'$. 
        Intuitively, the predicate $\new(\sigma,a)$ asserts whether there is an $a$-transition after the suspension trace $\sigma$ that is ``new'' with respect to the tested product $\lambda$.
    \item The set of transition relations $T'$ is defined as
        \[T'=\set{(\mcal X,a,\mcal Y)\in T\mid \mcal X,\mcal Y\in \X}.\]
    \item The set of product configurations $\Lambda' = \Lambda\setminus\set\lambda$.
\end{enumerate}
\end{definition}
Intuitively, Condition 1 defines $\X'$ to be a set of non-verdict states of the form $(X,\sigma)$ such that $\sigma$ is a suspension trace of the already tested product $\prd\lambda s$ and the predicate $\bt(X,\sigma)$ holds; whereas, $\X''$ is the set of non-verdict states reachable from a state in $\X'$ by a trace that is not a suspension trace of the tested product $\prd\lambda s$. Condition 2 and 3 are self-explanatory.

As an example, the spinal test suite generated from the test suite in Figure~\ref{fig:tsuite-cc} is partially drawn in Figure~\ref{fig:stsuite-cc}.
\begin{figure}[htb]
\centering
\tikzset{
  state/.style={
    rounded rectangle,
    draw=black
  }
}
{
\begin{tikzpicture}[->,shorten >=1pt,auto, semithick,inner sep=1pt, minimum size=0pt, every edge/.append style={font=\scriptsize},
every node/.append style={font=\scriptsize}]
\node[state] (s1) {$\set{s_0},\eseq$};
\node (s1a) at ($(s1.center)+(-2,0)$) {$\pass$};
\node (s1f) at ($(s1.center)+(3,0)$) {$\fail$};
\node[state] (s2) at ($(s1.center)+(0,-1.5)$) {$\set{s_1},\action{on}$};
\node (s2f) at ($(s2.center)+(3,0)$) {$\fail$};
\node (s2a) at ($(s2.center)+(-2,0)$) {$\pass$};
\node[state] (s3) at ($(s2.center)+(0,-1.5)$) {$\set{s_2},\action{on}\ \action{det}$};
\node (s3f) at ($(s3.center)+(3,0)$) {$\fail$};
\node (s3a) at ($(s3.center)+(-2,0)$) {$\cdots$};
\node (s3c) at ($(s3.center)+(0,-1.5)$) {$\cdots$};
\path[->]
    ($(s1.center)+(0,0.5)$) edge (s1)
    (s1) edge node[above] {\action{rgl,srgl}} (s1f)
    (s1) edge node[above] {\action{\delta}} (s1a)
    (s2) edge node[above] {\action{\delta,srgl}} (s2f)
    (s1) edge node {\action{on}} (s2)
    (s2) edge node[above] {\action{rgl}} (s2a)
    (s3) edge node[above] {\action{\delta,rgl}} (s3f)
    (s2) edge node {\action{det}} (s3)
    (s3) edge node[above right] {\action{nor}} (s3a)
    (s3) edge node[left] {\action{sgl}} (s3c);
\end{tikzpicture}
}
\caption{Spinal test suite of the cruise controller}
\label{fig:stsuite-cc}
\end{figure}
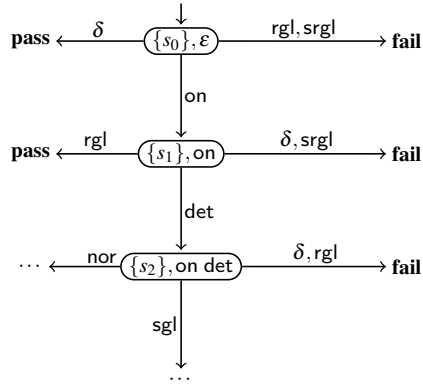

\section{Exhaustiveness of Spinal Test Suites}\label{sec:exhaust}
The spinal test suite $\sres$ contains the spines of those paths from the test suite $\mcal T(s,\phi)$ that lead to new behavior w.r.t.\ to the already-tested product $\lambda$.
Next, we show that  the spinal test suite $\sres$
is not necessarily exhaustive for an arbitrary implementation under test, i.e., it may have \emph{strictly} less testing power than the test suite $\mcal T(s,\phi)$. We show this through the following example.

\begin{example}
Consider an implementation of a cruise controller with a collision avoidance feature modeled as the IOFTS depicted in Figure~\ref{fig:cc-incorrect-impl}.
\begin{figure}[htb]
\centering
\tikzset{
  state/.style={
    circle,
    draw=black,
  }
}
\begin{tikzpicture}[->,shorten >=1pt,auto, semithick,inner sep=1pt, minimum size=0pt,every edge/.append style={font=\scriptsize}]
\node[state] (s1) {};
\node[state] (s2) at ($(s1.center) + (1,0)$) {};
\node[state] (s3) at ($(s2.center) + (0,-1)$) {};
\node[state] (s4) at ($(s3.center) + (2,0)$) {};
\node[state] (s5) at ($(s3.center) + (0,-1)$) {};
\node[state] (s6) at ($(s2.center) + (2,1)$) {};
\path[->]
    ($(s1.center)+(-0.5,0)$) edge (s1)
    (s1) edge node {\action{on}} (s2)
    (s2) edge node[right] {\action{off}} (s3)
    (s2) edge[loop above] node {\action{rgl}} (s2)
    (s3) edge[bend left] node {\action{on}} (s4)
    (s4) edge[loop right] node {\action{rgl}} (s4)
    (s4) edge[bend left] node {\action{off}} (s3)
    (s4) edge[bend left] node {\action{det}} (s5)
    (s5) edge[loop below] node {\action{rgl}} (s5)
    (s5) edge[bend left] node {\action{nor}} (s2)
    (s2) edge[bend left] node {\action{det}} (s6)
    (s6) edge[bend left] node {\action{nor}} (s2)
    (s6) edge[loop right] node {\action{srgl}} (s6);
\end{tikzpicture}
\caption{A faulty implementation of the cruise controller with control avoidance.}
\label{fig:cc-incorrect-impl}
\end{figure}
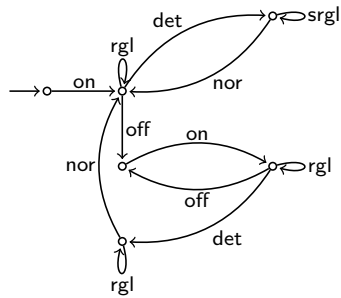
Clearly, this implementation is a faulty one as the action `\action{rgl}' must be prohibited after detecting an obstacle, i.e., after executing the transition labeled `\action{det}'.

As soon as we place the test suite (Figure~\ref{fig:tsuite-cc}) in parallel ($\parallel$) with the above-given implementation, we observe that the following synchronous interactions emerge: $\action{on}.\action{off}.\action{on}.\action{det}.\action{rgl}$, which lead to the $\fail$ verdict state. However, note that the aforementioned fault in the implementation cannot be detected while interacting with the spinal test suite of Figure~\ref{fig:stsuite-cc},  because there are no transitions labeled with \action{off} in the spinal test suite. Thus, a spinal test suite $\sres$ has strictly less testing power than the test suite $\mcal T(s,\phi)$.
\end{example}

Next, we explore when a spinal test suite $\sres$ (where $\lambda\models\phi$) together with a concrete test suite $\mcal T(s,\lambda)$ have the same testing power as the abstract test suite $\mcal T(s,\phi)$.
\begin{definition}
Let $\lambda\models\phi$. A feature specification $\prd{\phi'}{s'}$ is \emph{orthogonal} w.r.t $\prd \phi s$ and the product $\lambda$ iff 
$$\forall_{s_1,\sigma',a,\sigma''}\ \left(\new (\sigma',a) \wedge \prd{\phi'}{s'}\steps{\sigma' a\sigma''}\prd{\phi'}{s_1}\right) \ \Rightarrow \ \exists_{s_2,\sigma}\ \prd{\phi'}{s'}\steps{\sigma a\sigma''} \prd{\phi'}{s_2} \wedge \sigma\spine\sigma'.$$
\end{definition}

\begin{example}
Recall the feature specification $\prd \phi {s_0}$ and the product $\lambda$ (which omits the control avoidance feature) from Example~\ref{ex:phi-lambda}. Note that the implementation given in Figure~\ref{fig:cc-incorrect-impl} is not orthogonal w.r.t the feature specification $\prd \phi {s_0}$ and the product $\lambda$ because the underlined subsequence in ``$\action{on}\ \action{off}\ \action{on}\ \underline{\action{det} \ \action{rgl}}$'' cannot be extended with the spine sequence $\action{on}$.
\end{example}

In the remainder, we prove the main result (Theorem~\ref{thm:sres}) of this section that an orthogonal implementation passes the test suite $\mcal T(s,\phi)$ whenever it passes the concrete test suite $\mcal T(s,\lambda)$ and the spinal test suite $\sres$.
\begin{lemma}\label{lemma:spine-fail}
Let $\X_0$ be the initial state of a test suite $\mcal T(s,\phi)$ and let $\lambda$ be a product with $\lambda\models\phi$. If $\X_0 \steps {\sigma'a\sigma''}\fail$, $\new(\sigma',a)$, and $\sigma\spine\sigma'$ then $\X_0 \steps {\sigma a\sigma''}\fail$.
\end{lemma}
\begin{proof}[Proof sketch]
Let us first decompose the sequence of transitions $\X_0 \steps {\sigma'\sigma''}\fail$ as $\X_0 \steps {\sigma'} (X,\sigma') \steps {\sigma''}\fail$, for some $X$. Then by definition of a spine path we get $\iafter\phi\steps{\sigma} (X,\sigma)$. Next, it is straightforward to show by induction on $\sigma''$ that $(X,\sigma) \steps{a\sigma''} \fail$, whenever $(X,\sigma') \steps{a\sigma''} \fail$ and $\new (\sigma',a)$.
\end{proof}

\begin{theorem}\label{thm:sres}
Let $\prd {\phi'} {s'}$ be orthogonal w.r.t.\ to $\prd\phi s$ and $\lambda$. If $\prd {\phi'} {s'}$ pass the test suites $\mcal T(s,\lambda)$ and $\sres$, then $\prd {\phi'} {s'}$ passes the test suite $\mcal T(s,\phi)$.
\end{theorem}
\begin{proof}
Let $\X_0$ be the initial state of the test suite $\mcal T(s,\phi)$. We will prove this theorem by contradiction. Let $\prd {\phi'} {s'}$ pass the test suites $\mcal T(s,\lambda)$ and $\sres$. Suppose $\prd {\phi'} {s'}$ fails in passing the test suite $\mcal T(s,\phi)$. Then, there exists the following sequences of transitions $\X_0 \steps \sigma \fail$ and $\prd {\phi'}{s'} \steps \sigma \prd {\phi'}{s_1'}$ (for some $\sigma,s_1'$) in the test suite $\mcal T(s,\phi)$ and the feature specification $\prd {\phi'}{s'}$. Now there are two possibilities:
\begin{enumerate}
    \item Either, $\sigma\in\straces{\prd\lambda s}$. Then, the feature specification $\prd {\phi'}{s'}$ fails to pass the test suite $\mcal T(s,\lambda)$. Hence, a contradiction.
    \item Or, $\sigma\not\in\straces{\prd\lambda s}$. Then, the sequence of transitions $\X_0\steps\sigma\fail$ can be decomposed in the following way:
        $\X_0\steps {\sigma_1a\sigma_2} \fail$ with $\sigma=\sigma_1a\sigma_2$ and $\new (\sigma_1,a)$.

        Since the feature specification $\prd {\phi'}{s'}$ is orthogonal w.r.t.\ $\prd\phi s$ and $\lambda$, we have $$\exists_{s_2',\sigma_1'}\ \prd{\phi'}{s'} \steps {\sigma_1'a\sigma_2} \prd {\phi'}{s_2'} \wedge \sigma_1' \spine\sigma_1.$$ Then, by applying Lemma~\ref{lemma:spine-fail} we get the following path in the spinal test suite: $\X_0 \steps{\sigma_1'a\sigma_2} \fail$. Thus, $\prd{\phi'}{s'}$ fails to pass the spinal test suite $\sres$; hence, a contradiction.\qedhere
\end{enumerate}
\end{proof}

\section{Related work}\label{sec:related}

Various attempts have been made regarding formal and informal modeling of SPLs, on which \cite{Schaefer2012,Classen2010b,Schmid2011,Czarnecki:2012,Sinnema:2007} provide comprehensive surveys. By and large, the literature can be classified into two categories: structural modeling  and behavioral modeling techniques.

Structural models specify variability in terms of presence and absence of features (assets, artifacts) in various products and
their mutual inter-relations. Behavioral models, however, concern the working of features and their possible interactions, mostly based on some form of finite state machines or labeled transition systems.
The main focus in behavioral modeling of SPLs (cf. \cite{Asirelli:2011:compfts,Asirelli:2011:mcheck-services,Classen:2012:fts,Classen:2010:lots,Fischbein:2006:conformance,Gruler:2008:PL-ccs,Larsen:2007:MIA})
has been on formal specification of SPLs and adaptation of formal verification (mostly model checking) techniques to this new setting.


In addition, several testing techniques have been adapted to SPLs, of which \cite{Oster2011,Lamancha2013,Neto2011,Engstrom2011} provide recent overviews. Hitherto, most fundamental approaches to formal conformance testing \cite{book-mbt} have not been adapted sufficiently to the SPL setting. The only exception that we are aware of is \cite{inc-mbt-delta:2012}, which presents an LTS-based incremental derivation of test suites by applying principles of regression testing and delta-oriented modeling \cite{Abstract-delta:2010}.

Although our work is based on input-output conformance testing, we envisage that the ideas pursued in this paper can be adapted to other fundamental theories of conformance testing, e.g., those based on finite state machines \cite{book-mbt,Yannakakis99}.

\section{Conclusions}\label{sec:conc}
In this paper, we introduced the notion of spinal test suites, which can be used in order to incrementally test different products of a software product line.
A spinal test suite only tests the behavior induced by the ``new'' features and dispenses with re-testing the already-tested behavior, unless this is  necessary in order to reach the behavior of the new features.

As future work, we intend to exploit this notion and establish a methodology of testing software product lines, by automatically detecting the optimal order of testing products, which leads to a minimal size of residual test suites (with respect to a given notion of model coverage). In order to effectively use the notion of spinal test suites, we would like to define syntactic criteria that guarantee orthogonality of features.

\section*{Acknowledgments}
We thank the anonymous reviewers of MBT 2014 for their useful feedback.

\bibliographystyle{eptcs}
\bibliography{prodbib}

\begin{thebibliography}{10}
\providecommand{\bibitemdeclare}[2]{}
\providecommand{\surnamestart}{}
\providecommand{\surnameend}{}
\providecommand{\urlprefix}{Available at }
\providecommand{\url}[1]{\texttt{#1}}
\providecommand{\href}[2]{\texttt{#2}}
\providecommand{\urlalt}[2]{\href{#1}{#2}}
\providecommand{\doi}[1]{doi:\urlalt{http://dx.doi.org/#1}{#1}}
\providecommand{\bibinfo}[2]{#2}

\bibitemdeclare{inproceedings}{Asirelli:2011:mcheck-services}
\bibitem{Asirelli:2011:mcheck-services}
\bibinfo{author}{P.~\surnamestart Asirelli\surnameend}, \bibinfo{author}{M.~H.
  \surnamestart ter Beek\surnameend}, \bibinfo{author}{A.~\surnamestart
  Fantechi\surnameend} \& \bibinfo{author}{S.~\surnamestart Gnesi\surnameend}
  (\bibinfo{year}{2011}): \emph{\bibinfo{title}{A Model-Checking Tool for
  Families of Services}}.
\newblock In \bibinfo{editor}{R.~\surnamestart Bruni\surnameend} \&
  \bibinfo{editor}{J.~\surnamestart Dingel\surnameend}, editors: {\sl
  \bibinfo{booktitle}{Formal Techniques for Distributed Systems}}, {\sl
  \bibinfo{series}{Lecture Notes in Computer Science}} \bibinfo{volume}{6722},
  \bibinfo{publisher}{Springer Berlin Heidelberg}, pp. \bibinfo{pages}{44--58},
  \doi{10.1007/978-3-642-21461-5\_3}.

\bibitemdeclare{inproceedings}{Asirelli:2011:compfts}
\bibitem{Asirelli:2011:compfts}
\bibinfo{author}{P.~\surnamestart Asirelli\surnameend}, \bibinfo{author}{M.~H.
  \surnamestart ter Beek\surnameend}, \bibinfo{author}{S.~\surnamestart
  Gnesi\surnameend} \& \bibinfo{author}{A.~\surnamestart Fantechi\surnameend}
  (\bibinfo{year}{2011}): \emph{\bibinfo{title}{Formal Description of
  Variability in Product Families}}.
\newblock In \bibinfo{editor}{E.~\surnamestart Almeida\surnameend},
  \bibinfo{editor}{T.~\surnamestart Kishi\surnameend},
  \bibinfo{editor}{C.~\surnamestart Schwanninger\surnameend},
  \bibinfo{editor}{I.~\surnamestart John\surnameend} \&
  \bibinfo{editor}{K.~\surnamestart Schmid\surnameend}, editors: {\sl
  \bibinfo{booktitle}{Proc. of 15th International Software Product Line
  Conference}}, \bibinfo{publisher}{IEEE}, pp. \bibinfo{pages}{130--139},
  \doi{10.1109/SPLC.2011.34}.

\bibitemdeclare{inproceedings}{fioco-sac14}
\bibitem{fioco-sac14}
\bibinfo{author}{H.~\surnamestart Beohar\surnameend} \& \bibinfo{author}{M.~R.
  \surnamestart Mousavi\surnameend} (\bibinfo{year}{2014}):
  \emph{\bibinfo{title}{Input-Output Conformance Testing Based on Featured
  Transition Systems}}.
\newblock In: {\sl \bibinfo{booktitle}{Proceedings of the the 29th Symposium On
  Applied Computing}}, \bibinfo{publisher}{ACM Press}.
\newblock \bibinfo{note}{To appear, available from:
  {\texttt{http://ceres.hh.se/mediawiki/images/b/b0/}
  \texttt{Mousavi\_svt\_2014.pdf}}}.

\bibitemdeclare{book}{book-mbt}
\bibitem{book-mbt}
\bibinfo{author}{M.~\surnamestart Broy\surnameend},
  \bibinfo{author}{B.~\surnamestart Jonsson\surnameend}, \bibinfo{author}{J.-P.
  \surnamestart Katoen\surnameend}, \bibinfo{author}{M.~\surnamestart
  Leucker\surnameend} \& \bibinfo{author}{A.~\surnamestart
  Pretschner\surnameend} (\bibinfo{year}{2005}):
  \emph{\bibinfo{title}{Model-Based Testing of Reactive Systems}}.
\newblock {\sl \bibinfo{series}{Lecture Notes in Computer Science}}
  \bibinfo{volume}{3472}, \bibinfo{publisher}{Springer Berlin Heidelberg},
  \doi{10.1007/b137241}.

\bibitemdeclare{inproceedings}{Abstract-delta:2010}
\bibitem{Abstract-delta:2010}
\bibinfo{author}{D.~\surnamestart Clarke\surnameend},
  \bibinfo{author}{M.~\surnamestart Helvensteijn\surnameend} \&
  \bibinfo{author}{I.~\surnamestart Schaefer\surnameend}
  (\bibinfo{year}{2010}): \emph{\bibinfo{title}{Abstract delta modeling}}.
\newblock In \bibinfo{editor}{E.~\surnamestart Visser\surnameend} \&
  \bibinfo{editor}{J.~\surnamestart J\"{a}rvi\surnameend}, editors: {\sl
  \bibinfo{booktitle}{Proceedings of the 9th international conference on
  Generative programming and component engineering}}, \bibinfo{series}{GPCE
  '10}, \bibinfo{publisher}{ACM}, \bibinfo{address}{NY, USA}, pp.
  \bibinfo{pages}{13--22}, \doi{10.1145/1868294.1868298}.

\bibitemdeclare{techreport}{Classen2010b}
\bibitem{Classen2010b}
\bibinfo{author}{A.~\surnamestart Classen\surnameend} (\bibinfo{year}{2010}):
  \emph{\bibinfo{title}{Modelling with {FTS}: a Collection of Illustrative
  Examples}}.
\newblock \bibinfo{type}{Technical Report} \bibinfo{number}{P-CS-TR
  SPLMC-00000001}, \bibinfo{institution}{PReCISE Research Center, University of
  Namur}.
\newblock \urlprefix\url{http://www.fundp.ac.be/pdf/publications/69416.pdf}.

\bibitemdeclare{article}{Classen:2012:fts}
\bibitem{Classen:2012:fts}
\bibinfo{author}{A.~\surnamestart Classen\surnameend},
  \bibinfo{author}{M.~\surnamestart Cordy\surnameend}, \bibinfo{author}{P.-Y.
  \surnamestart Schobbens\surnameend}, \bibinfo{author}{P.~\surnamestart
  Heymans\surnameend}, \bibinfo{author}{A.~\surnamestart Legay\surnameend} \&
  \bibinfo{author}{J.-F. \surnamestart Raskin\surnameend}
  (\bibinfo{year}{2013}): \emph{\bibinfo{title}{Featured Transition Systems:
  Foundations for Verifying Variability-Intensive Systems and Their Application
  to LTL Model Checking}}.
\newblock {\sl \bibinfo{journal}{IEEE Transactions on Software Engineering}}
  \bibinfo{volume}{39}(\bibinfo{number}{8}), pp. \bibinfo{pages}{1069--1089},
  \doi{10.1109/TSE.2012.86}.

\bibitemdeclare{inproceedings}{Classen:2010:lots}
\bibitem{Classen:2010:lots}
\bibinfo{author}{A.~\surnamestart Classen\surnameend},
  \bibinfo{author}{P.~\surnamestart Heymans\surnameend}, \bibinfo{author}{P.-Y.
  \surnamestart Schobbens\surnameend}, \bibinfo{author}{A.~\surnamestart
  Legay\surnameend} \& \bibinfo{author}{J.-F. \surnamestart Raskin\surnameend}
  (\bibinfo{year}{2010}): \emph{\bibinfo{title}{Model Checking Lots of Systems:
  Efficient Verification of Temporal Properties in Software Product Lines}}.
\newblock In \bibinfo{editor}{J.~\surnamestart Kramer\surnameend},
  \bibinfo{editor}{J.~\surnamestart Bishop\surnameend}, \bibinfo{editor}{P.~T.
  \surnamestart Devanbu\surnameend} \& \bibinfo{editor}{S.~\surnamestart
  Uchitel\surnameend}, editors: {\sl \bibinfo{booktitle}{32nd International
  Conference on Software Engineering}}, {\sl \bibinfo{series}{ICSE
  '10}}~\bibinfo{volume}{1}, \bibinfo{publisher}{ACM}, \bibinfo{address}{New
  York, NY, USA}, pp. \bibinfo{pages}{335--344}, \doi{10.1145/1806799.1806850}.

\bibitemdeclare{inproceedings}{Czarnecki:2012}
\bibitem{Czarnecki:2012}
\bibinfo{author}{K.~\surnamestart Czarnecki\surnameend},
  \bibinfo{author}{P.~\surnamestart Gr{\"u}nbacher\surnameend},
  \bibinfo{author}{R.~\surnamestart Rabiser\surnameend},
  \bibinfo{author}{K.~\surnamestart Schmid\surnameend} \&
  \bibinfo{author}{A.~\surnamestart Wasowski\surnameend}
  (\bibinfo{year}{2012}): \emph{\bibinfo{title}{Cool Features and Tough
  Decisions: A Comparison of Variability Modeling Approaches}}.
\newblock In \bibinfo{editor}{U.~W. \surnamestart Eisenecker\surnameend},
  \bibinfo{editor}{S.~\surnamestart Apel\surnameend} \&
  \bibinfo{editor}{S.~\surnamestart Gnesi\surnameend}, editors: {\sl
  \bibinfo{booktitle}{Proceedings of the Sixth International Workshop on
  Variability Modeling of Software-Intensive Systems}}, \bibinfo{series}{VaMoS
  '12}, \bibinfo{publisher}{ACM}, \bibinfo{address}{New York, NY, USA}, pp.
  \bibinfo{pages}{173--182}, \doi{10.1145/2110147.2110167}.

\bibitemdeclare{article}{Engstrom2011}
\bibitem{Engstrom2011}
\bibinfo{author}{E.~\surnamestart Engstr{\"o}m\surnameend} \&
  \bibinfo{author}{P.~\surnamestart Runeson\surnameend} (\bibinfo{year}{2011}):
  \emph{\bibinfo{title}{Software Product Line Testing - A Systematic Mapping
  Study}}.
\newblock {\sl \bibinfo{journal}{Information {\&} Software Technology}}
  \bibinfo{volume}{53}(\bibinfo{number}{1}), pp. \bibinfo{pages}{2--13},
  \doi{10.1016/j.infsof.2010.05.011}.

\bibitemdeclare{inproceedings}{Fischbein:2006:conformance}
\bibitem{Fischbein:2006:conformance}
\bibinfo{author}{D.~\surnamestart Fischbein\surnameend},
  \bibinfo{author}{S.~\surnamestart Uchitel\surnameend} \&
  \bibinfo{author}{V.~\surnamestart Braberman\surnameend}
  (\bibinfo{year}{2006}): \emph{\bibinfo{title}{A Foundation for Behavioural
  Conformance in Software Product Line Architectures}}.
\newblock In \bibinfo{editor}{R.~M. \surnamestart Hierons\surnameend} \&
  \bibinfo{editor}{H.~\surnamestart Muccini\surnameend}, editors: {\sl
  \bibinfo{booktitle}{Proceedings of the ISSTA 2006 Workshop on Role of
  Software Architecture for Testing and Analysis}}, \bibinfo{series}{ROSATEA
  '06}, \bibinfo{publisher}{ACM}, \bibinfo{address}{New York, NY, USA}, pp.
  \bibinfo{pages}{39--48}, \doi{10.1145/1147249.1147254}.

\bibitemdeclare{incollection}{Gruler:2008:PL-ccs}
\bibitem{Gruler:2008:PL-ccs}
\bibinfo{author}{A.~\surnamestart Gruler\surnameend},
  \bibinfo{author}{M.~\surnamestart Leucker\surnameend} \&
  \bibinfo{author}{K.~\surnamestart Scheidemann\surnameend}
  (\bibinfo{year}{2008}): \emph{\bibinfo{title}{Modeling and Model Checking
  Software Product Lines}}.
\newblock In \bibinfo{editor}{G.~\surnamestart Barthe\surnameend} \&
  \bibinfo{editor}{F.~S. \surnamestart de~Boer\surnameend}, editors: {\sl
  \bibinfo{booktitle}{Proceedings of the 10th IFIP WG 6.1 International
  Conference on Formal Methods for Open Object-Based Distributed Systems}},
  {\sl \bibinfo{series}{Lecture Notes in Computer Science}}
  \bibinfo{volume}{5051}, \bibinfo{publisher}{Springer-Verlag},
  \bibinfo{address}{Berlin, Heidelberg}, pp. \bibinfo{pages}{113--131},
  \doi{10.1007/978-3-540-68863-1\_8}.

\bibitemdeclare{techreport}{Kang90}
\bibitem{Kang90}
\bibinfo{author}{K.~\surnamestart Kang\surnameend},
  \bibinfo{author}{S.~\surnamestart Cohen\surnameend},
  \bibinfo{author}{J.~\surnamestart Hess\surnameend},
  \bibinfo{author}{W.~\surnamestart Novak\surnameend} \&
  \bibinfo{author}{S.~\surnamestart Peterson\surnameend}
  (\bibinfo{year}{1990}): \emph{\bibinfo{title}{Feature-Oriented Domain
  Analysis ({FODA}) Feasibility Study}}.
\newblock \bibinfo{type}{Technical Report} \bibinfo{number}{CMU/SEI-90-TR-21},
  \bibinfo{institution}{Software Engineering Institute, Carnegie Mellon
  University}, \bibinfo{address}{Pittsburgh, Pennsylvania}.
\newblock
  \urlprefix\url{http://resources.sei.cmu.edu/library/asset-view.cfm?AssetID=1%
1231}.

\bibitemdeclare{inproceedings}{Lamancha2013}
\bibitem{Lamancha2013}
\bibinfo{author}{B.~P. \surnamestart Lamancha\surnameend},
  \bibinfo{author}{M.~\surnamestart Polo\surnameend} \&
  \bibinfo{author}{M.~\surnamestart Piattini\surnameend}
  (\bibinfo{year}{2013}): \emph{\bibinfo{title}{Systematic Review on Software
  Product Line Testing}}.
\newblock In \bibinfo{editor}{J.~\surnamestart Cordeiro\surnameend},
  \bibinfo{editor}{M.~\surnamestart Virvou\surnameend} \&
  \bibinfo{editor}{B.~\surnamestart Shishkov\surnameend}, editors: {\sl
  \bibinfo{booktitle}{Software and Data Technologies}}, {\sl
  \bibinfo{series}{Comm. in Computer and Information Science}}
  \bibinfo{volume}{170}, \bibinfo{publisher}{Springer Berlin Heidelberg}, pp.
  \bibinfo{pages}{58--71}, \doi{10.1007/978-3-642-29578-2\_4}.

\bibitemdeclare{techreport}{delangen12}
\bibitem{delangen12}
\bibinfo{author}{M.A.~de \surnamestart Langen\surnameend}
  (\bibinfo{year}{2013}): \emph{\bibinfo{title}{Vehicle Function Correctness}}.
\newblock \bibinfo{type}{Masters Thesis}, \bibinfo{institution}{Eindhoven
  University of Technology}.
\newblock
  \urlprefix\url{http://alexandria.tue.nl/extra1/afstversl/wsk-i/langen2013.pd%
f}.

\bibitemdeclare{incollection}{Larsen:2007:MIA}
\bibitem{Larsen:2007:MIA}
\bibinfo{author}{K.~G. \surnamestart Larsen\surnameend},
  \bibinfo{author}{U.~\surnamestart Nyman\surnameend} \&
  \bibinfo{author}{A.~\surnamestart W\k{a}sowski\surnameend}
  (\bibinfo{year}{2007}): \emph{\bibinfo{title}{Modal I/O Automata for
  Interface and Product Line Theories}}.
\newblock In: {\sl \bibinfo{booktitle}{Programming Languages and Systems}},
  {\sl \bibinfo{series}{Lecture Notes in Computer Science}}
  \bibinfo{volume}{4421}, \bibinfo{publisher}{Springer Berlin Heidelberg}, pp.
  \bibinfo{pages}{64--79}, \doi{10.1007/978-3-540-71316-6\_6}.

\bibitemdeclare{incollection}{inc-mbt-delta:2012}
\bibitem{inc-mbt-delta:2012}
\bibinfo{author}{M.~\surnamestart Lochau\surnameend},
  \bibinfo{author}{I.~\surnamestart Schaefer\surnameend},
  \bibinfo{author}{J.~\surnamestart Kamischke\surnameend} \&
  \bibinfo{author}{S.~\surnamestart Lity\surnameend} (\bibinfo{year}{2012}):
  \emph{\bibinfo{title}{Incremental Model-Based Testing of Delta-Oriented
  Software Product Lines}}.
\newblock In \bibinfo{editor}{A.~D. \surnamestart Brucker\surnameend} \&
  \bibinfo{editor}{J.~\surnamestart Julliand\surnameend}, editors: {\sl
  \bibinfo{booktitle}{Tests and Proofs}}, {\sl \bibinfo{series}{Lecture Notes
  in Computer Science}} \bibinfo{volume}{7305}, \bibinfo{publisher}{Springer
  Berlin Heidelberg}, pp. \bibinfo{pages}{67--82},
  \doi{10.1007/978-3-642-30473-6\_7}.

\bibitemdeclare{article}{Neto2011}
\bibitem{Neto2011}
\bibinfo{author}{Paulo~Anselmo \surnamestart da~Mota Silveira~Neto\surnameend},
  \bibinfo{author}{Ivan \surnamestart do~Carmo~Machado\surnameend},
  \bibinfo{author}{John~D. \surnamestart McGregor\surnameend},
  \bibinfo{author}{Eduardo~Santana \surnamestart de~Almeida\surnameend} \&
  \bibinfo{author}{Silvio~Romero \surnamestart de~Lemos~Meira\surnameend}
  (\bibinfo{year}{2011}): \emph{\bibinfo{title}{A systematic mapping study of
  software product lines testing}}.
\newblock {\sl \bibinfo{journal}{Information and Software Technology}}
  \bibinfo{volume}{53}(\bibinfo{number}{5}), pp. \bibinfo{pages}{407--423},
  \doi{10.1016/j.infsof.2010.12.003}.

\bibitemdeclare{incollection}{Oster2011}
\bibitem{Oster2011}
\bibinfo{author}{S.~\surnamestart Oster\surnameend},
  \bibinfo{author}{A.~\surnamestart W{\"{u}}bbeke\surnameend},
  \bibinfo{author}{G.~\surnamestart Engels\surnameend} \&
  \bibinfo{author}{A.~\surnamestart Sch{\"{u}}rr\surnameend}
  (\bibinfo{year}{2011}): \emph{\bibinfo{title}{A Survey of Model-Based
  Software Product Lines Testing}}.
\newblock In \bibinfo{editor}{J.~\surnamestart Zander\surnameend},
  \bibinfo{editor}{I.~\surnamestart Schieferdecker\surnameend} \&
  \bibinfo{editor}{P.~J. \surnamestart Mosterman\surnameend}, editors: {\sl
  \bibinfo{booktitle}{Model-based Testing for Embedded Systems}},
  \bibinfo{publisher}{CRC Press}, pp. \bibinfo{pages}{339--381},
  \doi{10.1201/b11321-14}.

\bibitemdeclare{article}{Schaefer2012}
\bibitem{Schaefer2012}
\bibinfo{author}{I.~\surnamestart Schaefer\surnameend},
  \bibinfo{author}{R.~\surnamestart Rabiser\surnameend},
  \bibinfo{author}{D.~\surnamestart Clarke\surnameend},
  \bibinfo{author}{L.~\surnamestart Bettini\surnameend},
  \bibinfo{author}{D.~\surnamestart Benavides\surnameend},
  \bibinfo{author}{G.~\surnamestart Botterweck\surnameend},
  \bibinfo{author}{A.~\surnamestart Pathak\surnameend},
  \bibinfo{author}{S.~\surnamestart Trujillo\surnameend} \&
  \bibinfo{author}{K.~\surnamestart Villela\surnameend} (\bibinfo{year}{2012}):
  \emph{\bibinfo{title}{Software diversity: state of the art and
  perspectives}}.
\newblock {\sl \bibinfo{journal}{International Journal on Software Tools for
  Technology Transfer}} \bibinfo{volume}{14}(\bibinfo{number}{5}), pp.
  \bibinfo{pages}{477--495}, \doi{10.1007/s10009-012-0253-y}.

\bibitemdeclare{inproceedings}{Schmid2011}
\bibitem{Schmid2011}
\bibinfo{author}{K.~\surnamestart Schmid\surnameend},
  \bibinfo{author}{R.~\surnamestart Rabiser\surnameend} \&
  \bibinfo{author}{P.~\surnamestart Gr{\"u}nbacher\surnameend}
  (\bibinfo{year}{2011}): \emph{\bibinfo{title}{A Comparison of Decision
  Modeling Approaches in Product Lines}}.
\newblock In \bibinfo{editor}{P.~\surnamestart Heymans\surnameend},
  \bibinfo{editor}{K.~\surnamestart Czarnecki\surnameend} \&
  \bibinfo{editor}{U.~W. \surnamestart Eisenecker\surnameend}, editors: {\sl
  \bibinfo{booktitle}{Proceedings of the 5th Workshop on Variability Modeling
  of Software-Intensive Systems}}, \bibinfo{series}{VaMoS '11},
  \bibinfo{publisher}{ACM}, \bibinfo{address}{New York, NY, USA}, pp.
  \bibinfo{pages}{119--126}, \doi{10.1145/1944892.1944907}.

\bibitemdeclare{inproceedings}{Schobbens:2006}
\bibitem{Schobbens:2006}
\bibinfo{author}{P.-Y. \surnamestart Schobbens\surnameend},
  \bibinfo{author}{P.~\surnamestart Heymans\surnameend} \&
  \bibinfo{author}{J.-C. \surnamestart Trigaux\surnameend}
  (\bibinfo{year}{2006}): \emph{\bibinfo{title}{Feature Diagrams: A Survey and
  a Formal Semantics}}.
\newblock In: {\sl \bibinfo{booktitle}{Proc. of the 14th IEEE International
  Conference on Requirements Engineering}}, \bibinfo{series}{RE '06},
  \bibinfo{publisher}{IEEE Computer Society}, \bibinfo{address}{Washington, DC,
  USA}, pp. \bibinfo{pages}{136--145}, \doi{10.1109/RE.2006.23}.

\bibitemdeclare{article}{Sinnema:2007}
\bibitem{Sinnema:2007}
\bibinfo{author}{M.~\surnamestart Sinnema\surnameend} \&
  \bibinfo{author}{S.~\surnamestart Deelstra\surnameend}
  (\bibinfo{year}{2007}): \emph{\bibinfo{title}{Classifying Variability
  Modeling Techniques}}.
\newblock {\sl \bibinfo{journal}{Information {\&} Software Technology}}
  \bibinfo{volume}{49}(\bibinfo{number}{7}), pp. \bibinfo{pages}{717--739},
  \doi{10.1016/j.infsof.2006.08.001}.

\bibitemdeclare{incollection}{Tretmans08}
\bibitem{Tretmans08}
\bibinfo{author}{J.~\surnamestart Tretmans\surnameend} (\bibinfo{year}{2008}):
  \emph{\bibinfo{title}{Model Based Testing with Labelled Transition Systems}}.
\newblock In \bibinfo{editor}{R.~M. \surnamestart Hierons\surnameend},
  \bibinfo{editor}{J.~P. \surnamestart Bowen\surnameend} \&
  \bibinfo{editor}{M.~\surnamestart Harman\surnameend}, editors: {\sl
  \bibinfo{booktitle}{Formal Methods and Testing}},
  chapter~\bibinfo{chapter}{I}, {\sl \bibinfo{series}{Lecture Notes in Computer
  Science}} \bibinfo{volume}{4949}, \bibinfo{publisher}{Springer Berlin
  Heidelberg}, pp. \bibinfo{pages}{1--38}, \doi{10.1007/978-3-540-78917-8\_1}.

\bibitemdeclare{inproceedings}{Yannakakis99}
\bibitem{Yannakakis99}
\bibinfo{author}{M.~\surnamestart Yannakakis\surnameend} \&
  \bibinfo{author}{D.~\surnamestart Lee\surnameend} (\bibinfo{year}{1999}):
  \emph{\bibinfo{title}{Testing of Finite State Systems}}.
\newblock In \bibinfo{editor}{G.~\surnamestart Gottlob\surnameend},
  \bibinfo{editor}{E.~\surnamestart Grandjean\surnameend} \&
  \bibinfo{editor}{K.~\surnamestart Seyr\surnameend}, editors: {\sl
  \bibinfo{booktitle}{Computer Science Logic}}, {\sl \bibinfo{series}{Lecture
  Notes in Computer Science}} \bibinfo{volume}{1584},
  \bibinfo{publisher}{Springer Berlin Heidelberg}, pp. \bibinfo{pages}{29--44},
  \doi{10.1007/10703163\_3}.

\end{thebibliography}

\end{document}